\documentclass[onecolumn,10pt]{IEEEtran}

\usepackage{graphicx}
\usepackage{amsmath,amssymb}
\usepackage{cases}
\usepackage{color}
\usepackage{amsfonts}
\usepackage{indentfirst}
\usepackage{slashbox}
\usepackage{cite}
\usepackage{caption}
\usepackage{algorithm}
\usepackage{algpseudocode}
\usepackage{booktabs}
\usepackage{subfigure}
\usepackage{multirow}
\usepackage{amsthm}
\makeatletter
\newtheoremstyle{mythm}{3pt}{3pt}{}{16pt}{\bfseries}{:}{.5em}{}
\theoremstyle{mythm}
\newtheorem{theorem}{Theorem}
\newtheorem{example}{Example}
\newtheorem{definition}{Definition}

\newtheorem{lemma}{Lemma}
\newtheorem{construction}{Construction}
\begin{document}
\title{A Novel Recursive Construction for Coded Caching Schemes
\author{Minquan Cheng, Jing Jiang, Youzhi Yao
}
\thanks{M. Cheng, J. Jiang and Y. Yao are with Guangxi Key Lab of Multi-source Information Mining $\&$ Security, Guangxi Normal University,
Guilin 541004, China (e-mail: $\{$chengqinshi,jjiang2008,yaoyzhi$\}$@hotmail.com).}
}
\date{}
\maketitle

\begin{abstract}
As a strategy to further reduce the transmission pressure during the peak traffic times in wireless network, coded caching has been widely studied recently. And several coded caching schemes are constructed focusing on the two core problems in practice, i.e., the rate transmitted during the peak traffic times and the packet number of each file divided during the off peak traffic times. It is well known that there exits a tradeoff between the rate and the packet number. In this paper, a novel recursive construction is proposed. As an application, several new schemes are obtained. Comparing with previously known schemes, new schemes could further reduce packet number by increasing little rate. And for some parameters in coded caching systems, the packet number of our new schemes are smaller than that of schemes generated by memory sharing method which is widely used in the field of caching. By the way our new schemes include all the results constructed by Tang et al., (IEEE ISIT, 2790-2794, 2017) as special cases.
\end{abstract}

\section{Introduction}
As an efficient solution to reduce tremendous pressure on the data transmission during the peak traffic times, caching has been widely studied \cite{AS,BGW,BG,DSS,GMDC,GGMG,JCM,KNMD,KPR} in heterogeneous wireless networks.
It has also been recognized as a disruptive technology to overcome the upcoming tremendous growth of wireless  traffics for the next 5th generation (5G) cellular networks \cite{BHLMP,WCTKL}. The basic idea is simple. During the off peak traffic times, some contents are proactively placed into the user's memory. Clearly if the content is required by user during the peak traffic times, then traffic amount can be reduced. So traditionally, most studies of caching system focused on exploiting the history or statistics of the user demands for an appropriate caching strategy, such as \cite{AA,BGW,BRS,DF,DSS,KPR,MMP}.

However, content requests are unknown to the server during the off peak traffic times.
Furthermore, the authors in \cite{STD} pointed that even for a known, fixed popularity distribution, deciding what to cache is an NP-hard combinatorial optimization problem.
Surprisingly even the contents which are required were not cached, we can also reduce the traffic amount by creating broadcast coding opportunities where the central server transmits the XOR of two files and each user uses their cache to cancel the non-desired file. This method is called coded caching first proposed by Maddah-Ali and Niesen in \cite{MN}. In \cite{MN}, the following caching scenario is focused: a single server containing $N$ files with the  same length connects to $K$ users over a shared link and each user has a cache memory of size $M$ files. An $F$-division coded caching scheme consists of two separate phases: placement phase and delivery phase.
In the placement phase, each file is divided into $F$ equal packets, and each user caches some packets of each file elaborately from server. This phase does not depend on the user demands
which are assumed to be arbitrary. In delivery phase, each user requires a file from server firstly. Then according to each user's cache, server sends a coded signal (XOR of some required packets) with size at most $R$ files to the users such that various user demands are satisfied. Here $R$ is always called the rate of the scheme.

The first $F$-division coded caching scheme was constructed by Maddah-Ali and Niesen in \cite{MN}. Such a  scheme is called MN scheme in this paper. In fact the rate of MN scheme is at most four times larger than the  information-theoretic lower bound on rate in \cite{GR}.
However $F$ in MN scheme increases exponentially with the number of users $K$. This would become infeasible when $K$ is large. Then many studies focus on designing caching schemes that decrease the order of $F$ for practical application. For example, a combinatorial structure, $(K,F,Z,S)$ placement delivery array (PDA) proposed in \cite{YCTC}, can be used to realize an $F$-division $(K,M,N)$ coded caching scheme where $M/N=Z/F$ and $R=S/F$. They also proved that MN scheme is equivalent to a special PDA which is denoted by MN PDA. There are other viewpoints of characterizing coded caching schemes, such as hypergraphs \cite{SZG}, resolvable designs and cyclic codes \cite{TR}, bipartite graphs \cite{YTCC}, Ruzsa-Szem\'{e}redi graphs \cite{STD} and so on. However \cite{SZG} showed that all the constructions in
\cite{SZG,TR,YTCC,STD} can be represented by PDAs. And the authors in \cite{CJYT} generalized all the constructions in \cite{YCTC} and most main results in \cite{SZG} by means of PDAs.

In this paper, we will propose an novel recursive construction of PDAs. Consequently several new schemes are obtained. By performance analyses, our new schemes can further reduce packet number by increase little rate comparing with the previously known results. And for some fixed $K$, $M/N$ and $R$, $F$ in our new schemes is smaller than that of schemes generated by memory sharing method. In addition, our results include all the results in \cite{TR} as a special case.

The rest of this paper is organized as follows. Section \ref{preliminaries} briefly reviews the relationship between PDA and coded caching scheme. In Section \ref{construction}, a recursive construction is proposed. In Section \ref{sec-application}, several classes of PDAs are obtained by using our recursive construction based on previously known PDAs. And some comparisons are considered. Conclusion is drawn in Section \ref{conclusion}.
\section{Preliminaries}\label{preliminaries}
\subsection{Placement Delivery Array}
\begin{definition}(\textit{Placement Delivery Array}, \cite{YCTC})
For  positive integers $K,F, Z$ and $S$, an $F\times K$ array  $\mathbf{P}=(p_{j,k})$, $0\leq j< F, 0\leq k< K$, composed of a specific symbol $``*"$  and $S$ integers
$0,1,\cdots, S-1$, is called a $(K,F,Z,S)$ placement delivery array (PDA) if it satisfies the following conditions:
\begin{enumerate}
  \item [C$1$.] The symbol $``*"$ appears $Z$ times in each column;
  \item [C2.] Each integer occurs at least once in the array;
  \item [C$3$.] For any two distinct entries $p_{j_1,k_1}$ and $p_{j_2,k_2}$,    $p_{j_1,k_1}=p_{j_2,k_2}=s$ is an integer  only if
  \begin{enumerate}
     \item [a.] $j_1\ne j_2$, $k_1\ne k_2$, i.e., they lie in distinct rows and distinct columns; and
     \item [b.] $p_{j_1,k_2}=p_{j_2,k_1}=*$, i.e., the corresponding $2\times 2$  subarray formed by rows $j_1,j_2$ and columns $k_1,k_2$ must be of the following form
  \begin{eqnarray*}
    \left(\begin{array}{cc}
      s & *\\
      * & s
    \end{array}\right)~\textrm{or}~
    \left(\begin{array}{cc}
      * & s\\
      s & *
    \end{array}\right).
  \end{eqnarray*}
   \end{enumerate}
\end{enumerate}
\end{definition}

Yan et al, in \cite{YCTC} showed that a $(K,F,Z,S)$ PDA $\mathbf{P}$ can be used to realize an $F$-division $(K,M,N)$ caching scheme with $M/N=Z/F$ and $R=S/F$.
\begin{theorem}(\cite{YCTC})
\label{th-Fundamental}An $F$-division caching scheme for a $(K,M,N)$ caching system can be realized by a $(K,F,Z,S)$ PDA  with $Z/F=M/N$. Each user can decode his requested file correctly for any request ${\bf d}$ at the rate $R=S/F$.
\end{theorem}
\begin{example}\label{exam2}
It is easy to verify that the following array is a $(4,6,3,4)$ PDA.
\begin{eqnarray*}
\mathbf{P}_{6\times 4}=\left(\begin{array}{cccc}
*&*&0&1\\
*&0&*&2\\
*&1&2&*\\
0&*&*&3\\
1&*&3&*\\
2&3&*&*
\end{array}\right).
\end{eqnarray*}
Then one can obtain a $6$-division $(4,3,6)$ coded caching scheme in the following way.
\begin{itemize}
   \item \textbf{Placement Phase}: First let $W_i$, $1 \leq i \leq 6$, be the $6$ files, and each file is divided into $F=6$ packets with equal size, i.e.,
   $$W_i=\{W_{i,0},W_{i,1},W_{i,2},W_{i,3},W_{i,4},W_{i,5}\},\ \ \ i\in [0,6).$$
  According to the positions of the stars in each column of $\mathbf{P}_{6\times 4}$, the contents cached in each users are defined as follows.
       \begin{align*}
       \mathcal{Z}_k=\left\{W_{i,j}\ |\ p_{j,k}=*, j\in [0,6), i\in[0,6)\right\},\ \  k\in [0,4)
       \end{align*}
  That is
       \begin{align*}
       \mathcal{Z}_0=\left\{W_{i,0},W_{i,1},W_{i,2}:i\in[0,6)\right\},\ \ \ \ \ \ \
       \mathcal{Z}_1=\left\{W_{i,0},W_{i,3},W_{i,4}:i\in[0,6)\right\}, \\
       \mathcal{Z}_2=\left\{W_{i,1},W_{i,3},W_{i,5}:i\in[0,6)\right\},\ \ \ \ \ \ \
       \mathcal{Z}_3=\left\{W_{i,2},W_{i,4},W_{i,5}:i\in[0,6)\right\}.
       \end{align*}
   \item \textbf{Delivery Phase}: Assume the $k$-th user requires the $k$-th file. For each integer $s\in [0,6)$, the server transmits the following coded signal in the $s$-th time.
       $$\oplus_{p_{j,k}=s, j\in[0,6),k\in[0,4)}W_{k,j}$$
  Table \ref{table1} lists all the coded signals.
   \begin{table}[!htp]
  \normalsize{
  \begin{tabular}{|c|c|}
\hline
    Time Slot& Transmitted Signnal  \\
\hline
   $0$&$W_{0,3}\oplus W_{1,1}\oplus W_{2,0}$\\ \hline
   $1$&$W_{0,4}\oplus W_{1,2}\oplus W_{3,0}$\\ \hline
  $2$& $W_{0,5}\oplus W_{2,2}\oplus W_{3,1}$\\ \hline
   $3$& $W_{1,5}\oplus W_{2,4}\oplus W_{3,3}$\\ \hline
  \end{tabular}}\centering
  \caption{Delivery steps in Example \ref{exam2} }\label{table1}
\end{table}
\end{itemize}
\end{example}
From Theorem \ref{th-Fundamental} and Example \ref{exam2}, we can obtain some coded caching schemes by constructing appropriate PDAs.

\subsection{Known results}
Here we list some previously known PDAs which have low $R$ and different level of packet numbers. For the other results the interested reader could be referred to \cite{CJYTC,SZG,YTCC,STD}.
\begin{lemma}(MN PDA\cite{MN})
\label{le-MN}
For any positive integers $K$ and $t$ with $t<K$, there exists a $(K,{K\choose t}, {K-1\choose t-1}, {K\choose t+1})$ PDA with $M/N=t/K$ and $R=\frac{K-t}{1+t}$.
\end{lemma}

\begin{lemma}(\cite{CJYT})\label{le-cheng2}
For any positive integers $q$, $z$ and $m$ with $q\geq2$ and $z<q$, there exists an $((m+1)q$, $\lfloor\frac{q-1}{q-z}\rfloor q^{m}$, $z\lfloor\frac{q-1}{q-z}\rfloor q^{m-1}$, $(q-z)q^{m})$ PDA with $M/N=\frac{z}{q}$ and $R=(q-z)/\lfloor\frac{q-1}{q-z}\rfloor$.
\end{lemma}
\begin{lemma}(\cite{CJYT})
\label{le-cheng-g1}
For any positive integers $q$, $z$, $m$ and $t$ with $q\geq2$, $z<q$ and $t<m$, there exists an $({m\choose t}q^t$, $\lfloor\frac{q-1}{q-z}\rfloor^t q^m$, $\lfloor\frac{q-1}{q-z}\rfloor^t(q^m-q^{m-t}(q-z)^t)$, $(q-z)^tq^{m})$ PDA with $M/N=1-(\frac{q-z}{q})^t$ and $R=(q-z)^t/\lfloor\frac{q-1}{q-z}\rfloor^t$.
\end{lemma}

In addition, the following result is very useful.
\begin{lemma}(\cite{CYTJ})
\label{le_permutations of PDA}
Let $\mathbf{P}$ be a $(K,F,Z,S)$ PDA for some positive integers $K$, $F$, $Z$ and $S$ with $Z<F$. Then there exists a $(K,S,S-(F-Z),F)$ PDA.
\end{lemma}

\section{Recursive constructions}
\label{construction}
For ease of introduction, we will use the following notations.
\begin{itemize}
\item ${\bf J}_{F\times K}$ denotes a matrix with $F$ rows and $K$ columns where each entry contains $1$. A ${\bf J}_{F\times 1}$ is denoted by ${\bf J}_F$ and ${\bf J}_F$ is written as ${\bf J}$ if there is no need to list the parameter $F$.
\item Given an array $\mathbf{P}=(p_{j,k})$, $0\leq j<F$, $0\leq k<K$ with alphabet $\{0,1,\ldots,S-1\}\bigcup\{*\}$, define $\mathbf{P}+a=(p_{j,k}+a)$ and $a\mathbf{P}=(ap_{j,k})$ for any integer $a$ where $*+a=*$ and $a*=*$.
\item $\langle a\rangle_b$ denotes the least nonnegative residue of $a$ modulo $b$ for any positive integers $a$ and $b$.
\end{itemize}
Now let us introduce our recursive construction. First the following constructions are very useful. Given two positive integers $u$ and $v$ with $u>v$ and $gcd(u,v)=1$, define a $u\times (u+v)$ array
\begin{small}
\begin{eqnarray}
\label{eq-lable-matrix}
\mathbf{A}=\left(\begin{array}{cccc|cccc }
0 & 1 & \ldots & u-1 & 0         & 1           & \ldots & v-1 \\
0 & 1 & \ldots & u-1 &v        & v+1       & \ldots & 2v-1\\
0 & 1 & \ldots & u-1 &2v       & 2v+1       & \ldots & 3v-1\\
  &   & \ldots &       &           &             & \ldots &          \\
0 & 1 & \ldots & u-1 &(u-1)v &(u-1)v+1 & \ldots & uv-1
\end{array}\right)
\end{eqnarray}
\end{small}
where all the operations are performed modulo $u$. For each $j$, $0\leq j<u$, denote
$$\mathcal{A}_j=\{\langle jv\rangle_u,\langle jv+1\rangle_u,\ldots,\langle (j+1)v-1\rangle_u\}.$$
We can define another $u\times (u+v)$ array $\mathbf{B} = (b_{j,k})$, where $0 \leq j < u$, $0 \leq k < u+v$, and each entry
\begin{eqnarray}\label{eq-recursive1}
b_{j,k}=\left\{
\begin{array}{ll}
j & \textrm{if}~k\in [0, u+v) \setminus \mathcal{A}_j\\
u+\lfloor \frac{jv+\langle k-jv\rangle_u}{u}\rfloor  & \textrm{if}~k\in \mathcal{A}_j\\
\end{array}
\right.
\end{eqnarray}
\begin{construction}
\label{con-recursive1}
For any $F\times u$ array $\mathbf{P}=({\bf p}_0,{\bf p}_1,\ldots,{\bf p}_{u-1})$ with $S$ distinct integers and any two matrices $\mathbf{A}=(a_{j,k})$, $\mathbf{B}=(b_{j,k})$, $0\leq j<u$, $0\leq k<u+v$, we can define a map $\psi$ from $(\mathbf{P}$, $\mathbf{A}$, $\mathbf{B})$ to an $(Fu)\times (u+v)$ array
\begin{small}
\begin{eqnarray}
\begin{split}
\label{eq-PDA-two-Ms}
\psi(\mathbf{P},\mathbf{A}, \mathbf{B})&=\left(\begin{array}{cccc}
{\bf p}_{a_{0,0}}  &{\bf p}_{a_{0,1}}  &\ldots&{\bf p}_{a_{0,u+v-1}}\\
{\bf p}_{a_{1,0}}  &{\bf p}_{a_{1,1}}  &\ldots&{\bf p}_{a_{1,u+v-1}}\\
\vdots             & \vdots            &\ddots&\vdots\\
{\bf p}_{a_{u-1,0}}&{\bf p}_{a_{u-1,1}}&\ldots&{\bf p}_{a_{u-1,u+v-1}}\\
\end{array}\right)+S\left(\begin{array}{cccc}
b_{0,0}{\bf J}&b_{0,1}{\bf J}&\ldots&b_{0,u+v-1}{\bf J}\\
b_{1,0}{\bf J}&b_{1,1}{\bf J}&\ldots&b_{1,u+v-1}{\bf J}\\
\vdots                    & \vdots                    &\ddots&\vdots\\
b_{u-1,0}{\bf J}&b_{u-1,1}{\bf J}&\ldots&b_{u-1,u+v-1}{\bf J}\\
\end{array}\right)\\[0.2cm]
&=\left(\begin{array}{cccc}
{\bf p}_{a_{0,0}}+b_{0,0}S{\bf J}    & {\bf p}_{a_{0,1}}+b_{0,1}S{\bf J}&\ldots&{\bf p}_{a_{0,u+v-1}}+b_{0,u+v-1}S{\bf J}\\
{\bf p}_{a_{1,0}}+b_{1,0}S{\bf J}    & {\bf p}_{a_{1,1}}+b_{1,1}S{\bf J}&\ldots&{\bf p}_{a_{1,u+v-1}}+b_{1,u+v-1}S{\bf J}\\
\vdots                               & \vdots                    &\ddots&\vdots\\
{\bf p}_{a_{u-1,0}}+b_{u-1,0}S{\bf J}&{\bf p}_{a_{u-1,1}}+b_{u-1,1}S{\bf J}&\ldots&{\bf p}_{a_{u-1,u+v-1}}+b_{u-1,u+v-1}S{\bf J}\\
\end{array}\right).
\end{split}
\end{eqnarray}
\end{small}
\end{construction}
\begin{example}
\label{ex-(5,9)PDA}
When $u=3$ and $v=2$, by \eqref{eq-lable-matrix} and \eqref{eq-recursive1} the following arrays can be obtained.
\begin{eqnarray}
\label{eq-lable-matrix-he}
\mathbf{A}=\left(\begin{array}{ccccc}
0 & 1 & 2 &0 & 1\\
0 & 1 & 2 &2 & 0\\
0 & 1 & 2 &1 & 2
\end{array}\right)\ \ \ \ \ \
\mathbf{B}=\left(\begin{array}{ccccc}
3 & 3 & 0 &0 & 0\\
4 & 1 & 3 &1 & 1\\
2 & 4 & 4 &2 & 2
\end{array}\right)
\end{eqnarray}
Let $F=3$ and
\begin{eqnarray*}
\mathbf{P}=({\bf p}_0,{\bf p}_1,{\bf p}_2)=\left(\begin{array}{ccc}
* & 0 & 1\\
0 & * & 2\\
1 & 2 & *
\end{array}\right).
\end{eqnarray*}
Then using $\mathbf{A}$ and $\mathbf{B}$ in \eqref{eq-lable-matrix-he}, the following array can be obtained by \eqref{eq-PDA-two-Ms},
\begin{small}
\begin{eqnarray*}
\psi(\mathbf{P},\mathbf{A}, \mathbf{B})=\left(\begin{array}{ccccc}
{\bf p}_0+3S{\bf J} &  {\bf p}_1+3S{\bf J}   & {\bf p}_2              & {\bf p}_0              &{\bf p}_1\\
{\bf p}_0+4S{\bf J} &  {\bf p}_1+S{\bf J}    & {\bf p}_2+3S{\bf J}  & {\bf p}_2+S{\bf J}   &{\bf p}_0+S{\bf J}\\
{\bf p}_0+2S{\bf J} &  {\bf p}_1+4S{\bf J}   & {\bf p}_2+4S{\bf J}  & {\bf p}_1+2S{\bf J}  &{\bf p}_2+2S{\bf J}
\end{array}\right)
=\left(\begin{array}{c|c|c|c|c}
*  & 9  & 1  & * & 0\\
9  & *  & 2  & 0 & *\\
10 & 11 & *  & 1 & 2\\
\hline
*  & 3  & 10 & 4 & *\\
12 & *  & 11 & 5 & 3\\
13 & 5  & *  & * & 4\\
\hline
*  & 12 & 13 & 6 & 7\\
6  & *  & 14 & * & 8\\
7  & 14 & *  & 8 & *\\
\end{array}\right).
\end{eqnarray*}
\end{small}
\end{example}
It is easy to check that $\mathbf{P}$ is a $(3,3,1,3)$ PDA and $\psi(\mathbf{P},\mathbf{A},\mathbf{B})$ is a $(5,9,3,15)$ PDA in Example \ref{ex-(5,9)PDA}.
\begin{lemma}
\label{le-gdc(K12)=1}
Given a $(K_1,F,Z,S)$ PDA, there exists a $(K_1+K_2,K_{1}F,K_{1}Z, (K_1+K_2)S)$ PDA for any positive integer $K_2$ with $gcd(K_1,K_2)=1$.
\end{lemma}
\begin{proof} Assume that $\mathbf{P}$ is  a $(K_1,F,Z,S)$ PDA. Let $u=K_1$ and $v=K_2$. Then matrices $\mathbf{A}$ and $\mathbf{B}$ can be obtained by \eqref{eq-lable-matrix} and \eqref{eq-recursive1} respectively. We claim that $\psi(\mathbf{P},\mathbf{A}, \mathbf{B})$ defined in \eqref{eq-PDA-two-Ms} is a $(u+v,uF,uZ, (u+v)S)$ PDA, i.e., $(K_1+K_2,K_{1}F,K_{1}Z, (K_1+K_2)S)$ PDA. It is easy to check that C$1$ and C$2$ hold by \eqref{eq-PDA-two-Ms}.
So we only need to consider C$3$. We consider two distinct sub-columns, say ${\bf p}_{a_{j,k}}+b_{j,k}S{\bf J}_F$ and ${\bf p}_{a_{j',k'}}+b_{j',k'}S{\bf J}_F$ with $j,j'\in [0, u)$ and $k,k'\in [0, u+v)$.
Clearly if $b_{j,k}\neq b_{j',k'}$, there is no integer both in ${\bf p}_{a_{j,k}}+b_{j,k}S{\bf J}_F$ and ${\bf p}_{a_{j',k'}}+b_{j',k'}S{\bf J}_F$.
So we only need to consider the case $b_{j,k}=b_{j',k'}$ with $(j, k) \neq (j', k')$.
\begin{itemize}
\item When $b_{j,k} = b_{j',k'}  \in [0,u)$, we have that $b_{j,k}=j$ and  $k, k' \in [0, u+v) \setminus \mathcal{A}_j$ from \eqref{eq-recursive1}. By \eqref{eq-PDA-two-Ms} we have
    $$\psi(\mathbf{P},\mathbf{A}, \mathbf{B})|_{b_{j,k}=j}=({\bf p}_{a_{j,k}})_{k\in [0, u+v) \setminus \mathcal{A}_j}+b_{j,k}S{\bf J}_{F\times u}=({\bf p}_{a_{j,k}})_{k\in [0, u+v) \setminus \mathcal{A}_j}+jS{\bf J}_{F\times u}.$$
    Clearly $({\bf p}_{a_{j,k}})_{k\in [0, u+v) \setminus \mathcal{A}_j}$ can be obtained by permutating the columns of $\mathbf{P}$.
    Hence it is also a PDA. So $\psi(\mathbf{P},\mathbf{A}, \mathbf{B})|_{b_{j,k}=j}$ satisfies C$3$, which implies
    ${\bf p}_{a_{j,k}}+b_{j,k}S{\bf J}$ and ${\bf p}_{a_{j,k}}+b_{j',k'}S{\bf J}$ satisfy C$3$ too.
\item When $b_{j,k} = b_{j',k'} \in [u,u+v)$, then $j,j'\in [0,u)$ and $k,k'\in [0,u)$ from \eqref{eq-recursive1}.
\begin{itemize}
\item If $j=j'$, $k\neq k'$ must hold by our hypothesis. Then we have the following sub-array
$$({\bf p}_{a_{j,k}}+b_{j,k}S{\bf J},{\bf p}_{a_{j,k'}}+b_{j,k'}S{\bf J})=({\bf p}_{a_{j,k}},{\bf p}_{a_{j,k'}})+b_{j,k}S({\bf J},{\bf J}).$$
Since $({\bf p}_{a_{j,k}},{\bf p}_{a_{j,k'}})$ is a sub-array of $\mathbf{P}$, $({\bf p}_{a_{j,k}},{\bf p}_{a_{j,k'}})$ satisfies C$3$. This implies that ${\bf p}_{a_{j,k}}+b_{j,k}S{\bf J}$ and ${\bf p}_{a_{j,k}}+b_{j,k'}S{\bf J}$ satisfy C$3$ too.
\item If $k=k'$, $j\neq j'$ must hold by our hypothesis.
    Since $b_{j,k}=b_{j',k}\in [u,u+v)$, we have $k\in \mathcal{A}_j \bigcap \mathcal{A}_{j'}$ from \eqref{eq-recursive1}, which implies $0 \leq k < u$. Without loss of generality, let $j'>j$,
$x_k=\langle k-jv\rangle_{u}$ and $x'_{k}=\langle k-{j'}v\rangle_{u}$.
Then $k+uz=jv+x_k$ and $k+uz'=j'v+x'_k$ for some nonnegative integers $z$ and $z'$. If $z'=z$, we have
  $\mathcal{A}_j\bigcap \mathcal{A}_{j'}=\emptyset$,
a contradiction to our hypothesis. So $z'>z$ always holds. Then $$b_{j,k}=u+\left\lfloor \frac{jv+x_k}{u}\right\rfloor=\left\lfloor z+\frac{k}{u}\right\rfloor=z \ \ \ \hbox{and}\ \ \ \
b_{j',k}=u+\left\lfloor \frac{j'v+x'_k}{u}\right\rfloor=\left\lfloor z'+\frac{k}{u}\right\rfloor=z'.
$$
This implies $b_{j,k}\neq b_{j',k}$, a contradiction to our assumption $b_{j,k}= b_{j',k}$. So this subcase does not happen.
\item If $j\neq j'$ and $k\neq k'$, the following sub-array can be obtained.
\begin{eqnarray*}
\left(\begin{array}{cc}
{\bf p}_{a_{j,k}}+b_{j,k}S{\bf J} &{\bf p}_{a_{j,k'}}+b_{j,k'}S{\bf J}\\
{\bf p}_{a_{j',k}}+b_{j',k}S{\bf J}&{\bf p}_{a_{j',k'}}+b_{j',k'}S{\bf J}
\end{array}\right)
\end{eqnarray*}
By \eqref{eq-lable-matrix} and \eqref{eq-PDA-two-Ms}, we have ${a_{j,k}}= {a_{j',k}}$ and ${a_{j,k'}}=a_{j',k'}$ since $k,k'\in [0,u)$. So together with our hypothesis of $b_{j,k}=b_{j',k'}$, the above sub-array can be written as follows.
\begin{eqnarray}\label{eq-1006-10-29}
\left(\begin{array}{cc}
{\bf p}_{a_{j,k}}+b_{j,k}S{\bf J} &{\bf p}_{a_{j',k'}}+b_{j,k'}S{\bf J}\\
{\bf p}_{a_{j,k}}+b_{j',k}S{\bf J}&{\bf p}_{a_{j',k'}}+b_{j,k}S{\bf J}
\end{array}\right)&=&\left(\begin{array}{cc}
{\bf p}_{a_{j,k}}&{\bf p}_{a_{j',k'}}\\
{\bf p}_{a_{j,k}}&{\bf p}_{a_{j',k'}}
\end{array}\right)+\left(\begin{array}{cc}
b_{j,k}S{\bf J} &b_{j,k'}S{\bf J}\\
b_{j',k}S{\bf J}&b_{j,k}S{\bf J}
\end{array}\right)
\end{eqnarray}
So the sub-array in \eqref{eq-1006-10-29} satisfies C$3$ if and only if the sub-array $({\bf p}_{a_{j,k}}, {\bf p}_{a_{j,k'}})$ satisfies C$3$. Clearly $({\bf p}_{a_{j,k}}, {\bf p}_{a_{j,k'}})$ satisfies C$3$ since it is a sub-array of $\mathbf{P}$.
\end{itemize}

\end{itemize}
\end{proof}

For  any two positive integers $K_1$ and $K_2$, let  $d = gcd(K_1,K_2)$,  $K'_1=K_1/d$ and $K'_2=K_2/d$.
Based on the $K'_1\times (K'_1+K'_2)$ matrices, $\mathbf{A}$ and $\mathbf{B}$ defined in \eqref{eq-lable-matrix} and \eqref{eq-recursive1} respectively, we can define another two $K'_1\times (K_1+K_2)$ matrices $\mathbf{A}'=({\bf a}'_{j,k})$ and $\mathbf{B}'=({\bf b}'_{j,k})$
where
$${\bf a}'_{j,k}=(da_{j,k},da_{j,k}+1, \ldots, da_{j,k}+d-1)\ \ \ \ \hbox{and}\ \ \ \ \ {\bf b}'_{j,k}=(b_{j,k},b_{j,k}, \ldots, b_{j,k})$$
are row vectors with length $d$, $0\leq j<K'_1$, $0\leq k<K'_1+K'_2$.
\begin{example}
\label{ex-DE-d=2}
When $K_1 = 6, K_2 = 4$, we have $d = gcd(K_1,K_2) = 2$, $K'_1=3$ and $K'_2=2$.
Then we can obtain two matrices $\mathbf{A}$  and $\mathbf{B}$ by applying \eqref{eq-lable-matrix} and \eqref{eq-recursive1}
with $u=K'_1=3$ and  $v=K'_2=2$ respectively, which were listed in \eqref{eq-lable-matrix-he}.
Based on these two matrices, the following two $3\times (6+4)$ matrices can be obtained.
\begin{eqnarray}
\label{eq-lable-matrix-ex-d=2}
\mathbf{A}'=\left(\begin{array}{cccccccccc}
0 & 1 &2 &3 &4 &5 &0 &1 &2 &3\\
0 & 1 &2 &3 &4 &5 &4 &5 &0 &1\\
0 & 1 &2 &3 &4 &5 &2 &3 &4 &5
\end{array}\right)\ \ \ \ \
\mathbf{B}'=\left(\begin{array}{cccccccccc}
3 &3 &3 &3 &0 &0 &0 &0 &0 &0\\
4 &4 &1 &1 &3 &3 &1 &1 &1 &1\\
2 &2 &4 &4 &4 &4 &2 &2 &2 &2
\end{array}\right)
\end{eqnarray}
\end{example}
Given a $(6,4,2,4)$ PDA
\begin{eqnarray*}
\mathbf{P}=\left(\begin{array}{cccccc}
*&*&*& 0& 1& 2\\
*& 0& 1&*&*& 3\\
0&*& 2&*& 3&*\\
1& 2&*& 3&*&*
\end{array}\right),
\end{eqnarray*}
and $\mathbf{A}'$, $\mathbf{B}'$ in \eqref{eq-lable-matrix-ex-d=2}, we have
\begin{small}
\begin{eqnarray*}
\psi(\mathbf{P},\mathbf{A}', \mathbf{B}')
=\left(\begin{array}{cccccc|cccc}
* &* &* &12&1 &2 &* &* &* &0\\
* &12&13&* &* &3 &* &0 &1 &*\\
12&* &14&* &3 &* &0 &* &2 &*\\
13&14&* &15&* &* &1 &2 &* &3\\
\hline
* &* &* &4 &13&14&5 &6 &* &*\\
* &16&5 &* &* &15&* &7 &* &4\\
16&* &6 &* &15&* &7 &* &4 &*\\
17&18&* &7 &* &* &* &* &5 &6\\
\hline
* &* &* &16&17&18&* &8 &9 &10\\
* &8 &17&* &* &19&9 &* &* &11\\
8 &* &18&* &19&* &10&* &11&*\\
9 &10&* &19&* &* &* &11&* &*
\end{array}\right).
\end{eqnarray*}
\end{small}
It is easy to check that the above array is a $(10,12,6,20)$ PDA. Similar to the proof of Lemma \ref{le-gdc(K12)=1}, we can also obtain a new PDA $\psi(\mathbf{P},\mathbf{A}', \mathbf{B}')$ defined in \eqref{eq-PDA-two-Ms} based on a given PDA $\mathbf{P}$ and the above matrices $\mathbf{A}'$, $\mathbf{B}'$.
\begin{theorem}
\label{th-fundamental recursive}
For any given $(K_1,F,Z,S)$ PDA and any positive integer $0<K_2\leq K_1$, there exists a $(K_{1}+K_{2},h_1F,h_1Z,(h_1+h_2)S)$ PDA where $h_1=K_{1}/gcd(K_{1},K_{2})$ and $h_2=K_{2}/gcd(K_{1},K_{2})$.
\end{theorem}
In fact, Theorem \ref{th-fundamental recursive} also holds when $K_1<K_2$ if we make small changes to the expression in \eqref{eq-recursive1}. However here we omit the similar introduction to save space. In the following section, let us consider the applications of Theorem  \ref{th-fundamental recursive} based on some previously known results.
\section{Applications}
\label{sec-application}
For simplicity, we only consider the case $K_2\leq K_1$. Let $K_1=dh_1$ and $K_2=dh_2$ where $d=gcd(K_1,K_2)$ and $h_1>h_2$ in the following.
\subsection{Application on MN scheme}
\label{subsec-app-MN}
From Lemma \ref{le-MN}, the following results can be obtained by Theorem \ref{th-fundamental recursive}.
\begin{theorem}
\label{th-MN-recursive}
For any positive integers $K_1$, $K_2$ and $t$ with $t,K_2<K_1$, there exists a $(K_1+K_2,h_1{K_1\choose t}, h_1{K_1-1\choose t-1}, (h_1+h_2){K_1\choose t+1})$ PDA with $M/N=t/K_1$ and $R=(1+\frac{h_2}{h_1})\frac{K_1-t}{1+t}$.
\end{theorem}

From Theorem \ref{th-MN-recursive}, we have a PDA with user number $K_1+K_2$, $M/N=t/K_1$ and
\begin{eqnarray}
\label{eq-com-MN-re_1}
F=h_1{K_1\choose t}, \ \ \ \ \ \ R=\left(1+\frac{h_2}{h_1}\right)\frac{K_1-t}{1+t}.
\end{eqnarray}

We claim that this PDA has smaller packet number but little larger rate than that of MN PDA.  For simplicity we assume that $t|K_2$. From Lemma \ref{le-MN}, we have an MN PDA with
\begin{eqnarray}
\label{eq-com-MN-re_2}
F_{MN}={K_{1}+K_{2}\choose (K_{1}+K_{2})t/K_{1}}, \ \ \ \ \ \ R_{MN}=\frac{(K_{1}+K_{2})\frac{K_{1}-t}{K_{1}}}{\frac{t(K_{1}+K_{2})}{K_{1}}+1}
=\frac{(K_{1}+K_{2})(K_{1}-t)}{t(K_{1}+K_{2})+K_{1}}.
\end{eqnarray}
So we have
\begin{eqnarray}
\label{eq-com-MN-re_3}
\begin{split}
&\frac{F}{F_{MN}}=\frac{h_1{K_{1}\choose t}}{{K_{1}+K_{2}\choose (K_{1}+K_{2})t/K_{1}}}\approx h_1 2^{-K_{2}H(\frac{t}{K_{1}})},\\
&\frac{R}{R_{MN}}=\left(1+\frac{h_2}{h_1}\right)\frac{K_{1}-t}{1+t}
\frac{t(K_{1}+K_{2})+K_{1}} {(K_{1}+K_{2})(K_{1}-t)}
=\left(1+\frac{h_2}{h_1}\right)\left(1-\frac{h_2}{(1+t)(h_1+h_2)}\right).
\end{split}
\end{eqnarray}
By the way the last approximate equality of the first formula holds for large enough $K_{1}$ where $H(\lambda)$ is the entropy function defined as follows.
$$H(\lambda)=-\lambda \log_2(\lambda)-(1-\lambda)\log_2(1-\lambda)\ \ \ \ 0<\lambda<1$$
For fixed $K_{1}$ and $M/N$, we have that $F$ is about $h^{-1}_1 2^{K_{2}H(\frac{t}{K_{1}})}$ times smaller than that of MN scheme, but $R$ just increases
 $$\left(1+\frac{h_2}{h_1}\right)\left(1-\frac{h_2}{(1+t)(h_1+h_2)}\right)<1+\frac{h_2}{h_1}<2$$ times. Clearly $h^{-1}_1 2^{K_{2}H(\frac{t}{K_{1}})}$ is increases exponentially with $K_{2}$, i.e., when $K_{1}$ and $K_{2}$ approach infinite, $h^{-1}_1  2^{K_{2}H(\frac{t}{K_{1}})}$ approach infinite very fast.
\begin{example}
When $M/N=0.5$ and $K_1/K_2=2$, the corresponding rates and packets numbers in the following table can be obtained by \eqref{eq-com-MN-re_1} and \eqref{eq-com-MN-re_2}.
\begin{eqnarray*}
\begin{array}{|c|c|c|c|c|c|c|}
\hline
K  &K_1 &K_2 &R      &R_{MN}&F          &  F_{MN}\\ \hline
6  &  4 &  2 &1.0000 &0.7500& 12        &20\\ \hline
12 &  8 &  4 &1.2000 &0.8571&       140 &       924\\ \hline
18 & 12 &  6 &1.2857 &0.9000&      1848 &     48620\\ \hline
24 & 16 &  8 &1.3333 &0.9231&     25740 &   2704156\\ \hline
30 & 20 & 10 &1.3636 &0.9375&    369512 & 155117520\\ \hline
36 & 24 & 12 &1.3846 &0.9474&   5408312 &9075135300\\ \hline
42 & 28 & 14 &1.4000 &0.9545&  80233200 &538257874440\\ \hline
48 & 32 & 16 &1.4118 &0.9600&1202160780 &32247603683100\\ \hline
\end{array}
\end{eqnarray*}
\end{example}

It is interesting that for the same $K$, $M/N$ and $R$, packet number of a PDA obtained by Theorem \ref{th-MN-recursive} is sometimes smaller than that of the scheme generated by memory sharing method based on MN PDAs. Let us introduce memory sharing method first. Given the parameters $K$ and $N$, assume that there exist $r$ schemes with $(M_1,R_1,F_1)$, $\ldots$, $(M_r,R_r,F_r)$ where $M_1<M_2<\ldots<M_r$. \cite{MN} showed that
a scheme with
\begin{eqnarray}
\label{eq-M-MNRF}
\begin{split}
\frac{M}{N}&=\lambda_1 \frac{M_1}{N}+\lambda_2 \frac{M_2}{N}+\ldots+\lambda_r \frac{M_r}{N}\\
R_{M-MN}&=\lambda_1 R_1+\lambda_2 R_2+\ldots+\lambda_r R_r\\
F_{M-MN}&=F_1+F_2+\ldots+F_r
\end{split}
\end{eqnarray}
can be obtained where $M_1\leq M\leq M_r$, $0< \lambda_i\leq 1$ and $\sum_{i=1}^{r}\lambda_i=1$.

Unfortunately it may be hard to theoretically compare with the advantages of them since the parameters are confusable. Instead we illustrate the reduction in packet number by numerical comparisons. Given a positive integer $K$, there exist $K+1$ MN PDAs with $M=0,N/K,\ldots,N(K-1)/K,N$ respectively
from Lemma \ref{le-MN}. Based on these MN PDAs, we have many pairs $(R_{M-MN},F_{M-MN})$ obtained by \eqref{eq-M-MNRF} for a fixed $M/N$. Denote the set of all these pairs by $\mathcal{H}$. For the fixed $M/N$, we can choose appropriate positive integers $K_1$ and $K_2$ such that $K = K_1 + K_2$ and obtain a PDA with $F$ and $R$ from Theorem \ref{th-MN-recursive}. With the aid of a computer we can find out all the pairs from $\mathcal{H}$, say $(R_{M-MN}^{(1)},F_{M-MN}^{(1)})$, $(R_{M-MN}^{(2)},F_{M-MN}^{(2)})$, $\ldots$, $(R_{M-MN}^{(s)},,F_{M-MN}^{(s)})$, such that
$$|R_{M-MN}^{(1)}-R|=\ldots=|R_{M-MN}^{(s)}-R|=\min\{|R_{M-MN}-R|\ |\ (R_{M-MN},F_{M-MN})\in \mathcal{H}\}$$
With out loss of generality we assume that $F^{(1)}_{M-MN} = \min\{F_{M-MN}^{(i)} \ | \ 1 \leq i \leq s\}$.
From Table \ref{tab_com-MN-MNrecursive-R}, the packet number obtained by Theorem \ref{th-MN-recursive} is smaller even though $R\leq R^{(1)}_{M-MN}$ for the same $K$ and $M/N$.
\begin{small}
\begin{table}[H]
  \centering
\caption{Schemes generated by Theorem \ref{th-MN-recursive} and memory sharing method based on MN PDAs} \label{tab_com-MN-MNrecursive-R}
  \normalsize{
  \begin{tabular}{|c|c|c|c|c|c|c|c|c|c|c|}
\hline
$K$&$K_1$&$K_2$&$h_1$&$h_2$&$t$&$M/N$&$R$&$R^{(1)}_{M-MN}$&$F$&$F^{(1)}_{M-MN}$\\ \hline
18 & 12 & 6 & 2 & 1 & 9 & 3/4 & 0.450 & 0.484 & 440 & 816\\ \hline
20 & 11 & 9 & 11 & 9 & 7 & 7/11 & 0.909 & 0.928 & 3630 & 5035\\ \hline
25 & 14 & 11 & 14 & 11 & 11 & 11/14 & 0.446 & 0.454 & 5096 & 12675\\ \hline
28 & 15 & 13 & 15 & 13 & 11 & 11/15 & 0.622 & 0.655 & 20475 & 101556\\ \hline
30 & 18 & 12 & 3 & 2 & 11 & 11/18 & 0.972 & 0.993 & 95472 & 593776\\ \hline
34 & 18 & 16 & 9 & 8 & 10 & 5/9 & 1.373 & 1.402 & 393822 & 5657872\\ \hline
36 & 19 & 17 & 19 & 17 & 13 & 13/19 & 0.812 & 0.851 & 515508 & 8724672\\ \hline
40 & 21 & 19 & 21 & 19 & 14 & 2/3 & 0.889 & 0.922 & 2441880 & 80743065\\ \hline
44 & 24 & 20 & 6 & 5 & 13 & 13/24 & 1.440 & 1.460 & 14976864 & 710016516\\ \hline
48 & 25 & 23 & 25 & 23 & 16 & 16/25 & 1.016 & 1.032 & 51074375 & 6918064890\\ \hline
50 & 26 & 24 & 13 & 12 & 18 & 9/13 & 0.810 & 0.815 & 20309575 & 10809156820\\ \hline
\end{tabular}}
\end{table}
\end{small}

Finally let us consider the performance comparing with the PDAs in Lemma \ref{le-cheng2}. The authors in \cite{CJYT} showed that PDAs in Lemma \ref{le-cheng2} have lower packet number level than that of MN PDAs. By the way when $M/N=1/q$ and $1-1/q$, this fact had also been shown in \cite{TR,YCTC} respectively. Given a positive integer $K$, we have many types of PDAs by choosing different $q$ and $z$ for a fixed $M/N$. Based on these PDAs, we have many pairs $(R_{M-le2},F_{M-le2})$ obtained by \eqref{eq-M-MNRF} for a fixed $M/N$. Similarly we can also obtain a PDA with $R$ and $F$ from Theorem \ref{th-MN-recursive}. And we also assume that $R^{(1)}_{M-le2}$ is nearest $R$ and $F^{(1)}_{M-le2}$ is the related minimum packet number. For some fixed $K$, $M/N$, three types of schemes are listed in Table \ref{tab_com-le2-MN-MNrecursive-R}. From Table \ref{tab_com-le2-MN-MNrecursive-R}, we have $R<R^{(1)}_{M-le2}$,  $R<R^{(1)}_{M-MN}$ and $F^{(1)}_{M-le2}<F<F^{(1)}_{M-MN}$. This implies that comparing with memory sharing method based on MN PDAs and PDAs in Lemma \ref{le-cheng2} respectively, Theorem \ref{th-MN-recursive} has smaller packet number but larger rate, and Theorem \ref{th-MN-recursive} has larger packet number but smaller rate.
\begin{small}
\begin{table}[H]
  \centering
\caption{PDAs in Theorem \ref{th-MN-recursive} and generated by memory sharing based on MN PDAs and PDAs in Lemma \ref{le-cheng2}} \label{tab_com-le2-MN-MNrecursive-R}
  \normalsize{
  \begin{tabular}{|c|c|c|c|c|c|c|c|c|c|c|c|c|}
\hline
$K$&$K_1$&$K_2$&$t$&$q$&$m$&$M/N$&$R^{(1)}_{M-le2}$&$R$&$R^{(1)}_{M-MN}$&$F^{(1)}_{M-le2}$&$F$&$F^{(1)}_{M-MN}$\\ \hline
18&12&6 &9 &6 &2 &3/4  &0.600&0.450&0.466&252    &440     &43776\\ \hline
20&11&9 &6 &5 &3 &6/11 &1.401&1.299&1.351&375    &5082    &16644\\ \hline
28&15&13&8 &4 &6 &8/15 &1.489&1.452&1.461&16384  &96525   &1560780\\ \hline
35&20&15&11&5 &6 &11/20&1.375&1.313&1.327&46875  &671840  &183631756\\ \hline
38&20&18&11&19&1 &11/20&4.125&1.425&1.441&95     &1679600 &472807571\\ \hline
44&24&20&13&4 &10&13/24&1.444&1.440&1.455&4194304&14976864&686354883984\\ \hline
\end{tabular}}
\end{table}
\end{small}
Similarly we can also discuss the performance of Theorem \ref{th-MN-recursive} comparing with other PDAs in Lemmas \ref{le-cheng-g1}.
\subsection{Applications on other known PDAs}
Based on the PDAs in Lemmas \ref{le-cheng2} and \ref{le-cheng-g1}, we can obtain the following results by Theorem \ref{th-fundamental recursive}.
\begin{theorem}\label{th-cheng2-re}
Given positive integers $q$, $z$, $m$ with $q\geq2$ and $z<q$, there exists an $((m+1)q+K_2$, $h_1\lfloor\frac{q-1}{q-z}\rfloor q^{m}$, $h_1z\lfloor\frac{q-1}{q-z}\rfloor q^{m-1}$, $(h_1+h_2)(q-z)q^{m})$ PDA with $M/N=z/q$ and $R=(1+\frac{h_2}{h_1})\frac{q-z}{\lfloor\frac{q-1}{q-z}\rfloor}$.
\end{theorem}

\begin{theorem}
\label{th-cheng-g1-re}
For any positive integers $q$, $z$, $m$ and $t$ with $q\geq2$, $z<q$ and $t<m$, there exists an $({m\choose t}q^t + K_2$, $h_1\lfloor\frac{q-1}{q-z}\rfloor^t q^m$, $h_1\lfloor\frac{q-1}{q-z}\rfloor^t(q^m-q^{m-t}(q-z)^t)$, $(h_1+h_2)(q-z)^tq^{m})$ PDA with $M/N=1-(\frac{q-z}{q})^t$ and $R=(1+\frac{h_2}{h_1})(q-z)^t/\lfloor\frac{q-1}{q-z}\rfloor^t$.
\end{theorem}

Similarly let us conder PDAs in Theorem \ref{th-cheng2-re} and in Lemma \ref{le-cheng2} for some parameters $K$. Given any positive integers $q$, $m$, let $K_1=(m+1)q$ and $K_2=qm'<K_1$ for simplicity, where $m'$ is a positive integer. When $M/N=z/q$, from Theorem \ref{th-cheng2-re}, we have a PDA with user number $K_1+K_2$ and
\begin{eqnarray}
\label{eq-com-cheng-re_1}
F=h_1 \lfloor\frac{q-1}{q-z}\rfloor q^m\ \ \ \ \ \ R=(1+\frac{h_2}{h_1})\frac{q-z}{\lfloor\frac{q-1}{q-z}\rfloor}
\end{eqnarray}
where
\begin{eqnarray}
\label{eq-com-cheng-re_1/q}
h_1=\frac{K_1}{gcd(K_1,K_2)}=\frac{m+1}{gcd(m+1,m')},\ \  \ \ h_2=\frac{K_2}{gcd(K_1,K_2)}=\frac{m'}{gcd(m+1,m')}.
\end{eqnarray}
Now let us consider the decrement of packet number. When $K=K_1+K_2=q(m+m'+1)$ and $M/N=z/q$, from Lemma \ref{le-cheng2}, we have a PDA with
\begin{eqnarray*}
F_{le2}=\lfloor\frac{q-1}{q-z}\rfloor q^{m+m'}\ \ \ \ \ \ R_{le2}=\frac{q-z}{\lfloor\frac{q-1}{q-z}\rfloor}\label{eq-com-cheng-re_2}
\end{eqnarray*}
So we have
\begin{eqnarray}
\label{eq-com-cheng-re_3}
\begin{split}
\frac{F}{F_{le2}}=h_1 q^{-m'}\ \ \ \ \ \frac{R}{R_{le2}}=1+\frac{h_2}{h_1}=1+\frac{m'}{m+1}=\frac{m+1+m'}{m+1}<2.
\end{split}
\end{eqnarray}
For  fixed $K$ and $M/N$, we have that $F$ is exactly $h^{-1}_1 q^{m'}$ times smaller than that of a PDA in Lemma \ref{le-cheng2}, but $R$ just increases $1+\frac{h_2}{h_1}<2$ times.

In addition, the results in Claims 3, 4 of \cite{TR}, which have been summarized in Theorem 1 of \cite{TR1}, are just our special cases. First let us introduce the results in Theorem 1 of \cite{TR1}. For any positive integers $n$, let $K=nq$ and $x$ be the least positive integer such that $(m+1)|nx$. References \cite{TR,TR1} proposed the following coded caching schemes with $K=nq$ when {\em $q$ is a prime power and satisfies some conditions related the $[n,m]$ cyclic code} (The interested readers could be referred to \cite{TR1} for more details).
\begin{itemize}
\item When $M/N=1/q$, there exists a coded caching scheme with
\begin{eqnarray}
\label{eq-Li-1}
R_{1}=\frac{n}{m+1}(q-1)\ \ \hbox{and} \ \ \ F_{1}=q^m x;
\end{eqnarray}

\item When $M/N=1-\frac{m+1}{nq}$, there exists a coded caching scheme with
\begin{eqnarray}
\label{eq-Li-2}
R_{2}=\frac{m+1}{(q-1)n}\ \ \hbox{and} \ \ \ F_{2}=(q-1)q^m \frac{xn}{m+1}.
\end{eqnarray}
\end{itemize}
Let us consider the case $M/N=1/q$ first. From Theorem \ref{th-cheng2-re} we have an $((m+1)q+K_2$, $h_1q^{m}$, $h_1q^{m-1}$, $(h_1+h_2)(q-1)q^{m})$ PDA with $M/N=1/q$ and $R=(1+\frac{h_2}{h_1})(q-1)$. Let $K_1=(m+1)q$ and $K=K_1+K_2$. Clearly $q|K_2$. Then we have
$$n=\frac{K_1+K_2}{q}=\frac{(h_1+h_2)gcd(K_1,K_2)}{q}.$$
By \eqref{eq-com-cheng-re_1/q}, $h_1=x$ holds and the following equations can be easily verified.
\begin{eqnarray*}
&&R=(1+\frac{h_2}{h_1})(q-1)=\frac{(h_1+h_2)\cdot gcd(K_1,K_2)}{h_{1}\cdot gcd(K_1,K_2)}(q-1)=\frac{nq}{K_1}=\frac{n}{m+1}(q-1)=R_{1}\\[0.2cm]
&&F=h_1q^{m}=xq^{m}=F_{1}
\end{eqnarray*}

Now let us consider the case $M/N=1-\frac{k+1}{nq}$.
With the same parameter $K_1$ and $K_2$ in the case $M/N=1/q$, from Theorem \ref{th-cheng2-re} and Lemma \ref{le_permutations of PDA},
we have an $((m+1)q+K_2, (h_1+h_2)(q-1)q^{m},(h_1+h_2)(q-1)q^{m}- (h_1q^{m}-h_1q^{m-1}),h_1q^{m})$ PDA with
\begin{eqnarray}
\label{eq-ch1/qMN}
\frac{M}{N}=\frac{(h_1+h_2)(q-1)q^{m}- (h_1q^{m}-h_1q^{m-1})}{(h_1+h_2)(q-1)q^{m}}
=1-\frac{h_1}{(h_1+h_2)q}=1-\frac{m+1}{nq},
\end{eqnarray}
\begin{eqnarray}
\label{eq-ch1/nqR}
R=\frac{h_1q^{m}}{(h_1+h_2)(q-1)q^{m}}=\frac{h_1}{(h_1+h_2)(q-1)}=\frac{m+1}{(q-1)n}=R_2
\end{eqnarray}
and
\begin{eqnarray}
\label{eq-ch1/nqF}
F=(h_1+h_2)(q-1)q^{m}=\frac{nx}{m+1}(q-1)q^{m}=F_2
\end{eqnarray}
where the second item in \eqref{eq-ch1/nqF} can be derived by the fact
$$(h_1+h_2)=(h_1+h_2)\cdot\frac{gcd(K_1,K_2)}{gcd(K_1,K_2)}\cdot\frac{h_1}{h_1}=\frac{nq h_1}{(m+1)q}=\frac{nx}{m+1}.$$

It is worth to note that with the same parameters $K$, $M/N$, $R$ and $F$ in \cite{TR}, our result holds for {\em any positive integer} $q\geq 2$. We can discuss the PDAs in Theorem  \ref{th-cheng-g1-re} similarly.

\section{Conclusion}
\label{conclusion}
In this paper, a novel recursive construction was proposed. As an application, several new schemes were obtained. Comparing with previously known schemes, our new schemes could further reduce packet number by increasing little rate. And the packet number of our new schemes are smaller than that of schemes generated by memory sharing method for some fixed $K$, $M/N$ and $R$. In addition, our new schemes include all the results in \cite{TR} as special cases.

\end{document}